\documentclass[floatfix,showpacs,twocolumn,longbibliography]{revtex4-1}

\usepackage{times,amsmath,amsfonts,amssymb,latexsym,amsthm,mathtools,natbib,multirow,epsf,latexsym,setspace,array}
\usepackage{graphicx,graphics,epsfig,epsf}
\usepackage[normalem]{ulem}
\usepackage{bm,,bbm,mathrsfs,color,xcolor}
\usepackage{hypernat,hyperref}
\usepackage{dsfont}
\usepackage{tikz}
\usetikzlibrary{matrix,arrows,shapes,backgrounds,patterns,calc}
\usepackage{tikz-cd}
\usepackage{pgfplots}
\usepackage{soul} 
\usetikzlibrary{external}
\usepackage{subfigure}

\newtheorem{mylemma}{Lemma}
\newtheorem{mydefinition}{Definition}
\newtheorem{mytheorem}{Theorem}
\newtheorem{myremark}{Remark}
\newtheorem{mycorollary}{Corollary}
\newtheorem{myexample}{Example}

\newcommand{\ket}[1]{\left| #1 \right\rangle}

\newcommand{\ketbra}[2]{| #1 \rangle\! \langle #2 |}

\newcommand{\ignore}[1]{}

\DeclareMathOperator{\Tr}{Tr}
\DeclareMathOperator{\Ad}{Ad}
\DeclareMathOperator{\Span}{Span}
\DeclareMathOperator{\id}{id}
\newcommand{\TrB}{\Tr_{\textsc{b}}}

\newcommand{\TrW}{\Tr_{\textsc{w}}}
\newcommand{\identity}{\mathds{1}}
\newcommand{\Sgp}{\mathcal{G}} 
\newcommand{\HH}{\mathcal{H}}  
\newcommand{\HS}{\mathcal{H}_{\textsc{s}}}
\newcommand{\HB}{\mathcal{H}_{\textsc{b}}}
\newcommand{\HW}{\mathcal{H}_{\textsc{w}}}
\newcommand{\HSB}{\HS\otimes\HB}

\newcommand{\BL}{\mathcal{B}}
\newcommand{\BH}{\BL(\HH)}
\newcommand{\BHS}{\BL(\HS)}

\newcommand{\BHW}{\BL(\HW)}
\newcommand{\BHSB}{\BL(\HSB)}

\newcommand{\UW}{\mathrm{U}(\HW)}
\newcommand{\USB}{\mathrm{U}(\HSB)}

\newcommand{\DD}[1]{\mathcal{D}_{\textsc{#1}}}
\newcommand{\DS}{\DD{s}}
\newcommand{\DB}{\DD{b}}

\newcommand{\DSB}{\DD{sb}}
\newcommand{\DSW}{\DD{sw}}
\newcommand{\DBW}{\DD{bw}}
\newcommand{\DSBW}{\DD{sbw}}
\newcommand{\rhoS}{\rho_{\textsc{s}}}
\newcommand{\rhoB}{\rho_{\textsc{b}}}
\newcommand{\rhoSB}{\rho_{\textsc{sb}}}
\newcommand{\rhoW}{\rho_{\textsc{w}}}

\newcommand{\rhoBW}{\rho_{\textsc{bw}}}
\newcommand{\rhoSBW}{\rho_{\textsc{sbw}}}

\newcommand{\CS}{\mathcal{V}} 
\newcommand{\rmT}{\mathrm{T}}
\newcommand{\jmdotimes}{\!\otimes\!}

\begin{document}

\title{Beyond Complete Positivity}
\author{Jason M. Dominy$^{(1,4)}$ and Daniel A. Lidar$^{(1,2,3,4)}$}
\affiliation{Departments of $^{(1)}$Chemistry, $^{(2)}$Physics, $^{(3)}$Electrical Engineering, $^{(4)}$Center for Quantum Information Science \& Technology, University of Southern California, Los Angeles, CA 90089, USA}

\begin{abstract}
We provide a general and consistent formulation for linear subsystem quantum dynamical maps, developed from a minimal set of 
postulates, primary among which is a relaxation of the usual, restrictive assumption of uncorrelated initial system-bath states.  We describe the space of possibilities admitted by this formulation, namely that, far from being limited to only completely positive (CP) maps, essentially any $\mathbb{C}$-linear, Hermiticity-preserving, 
trace-preserving 
subsystem map 
can arise as a legitimate subsystem dynamical map from a joint unitary evolution of a system coupled to a bath.  The price paid for this added generality is a trade-off between the set of admissible initial states and the allowed set of joint system-bath unitary evolutions. 
As an application we present a simple example of a non-CP map constructed as a subsystem dynamical map that violates some fundamental inequalities in quantum information theory, such as the quantum data processing inequality. 
\end{abstract}
\maketitle

\textit{Introduction}.---%
The theory of open quantum systems deals with quantum subsystems that interact with an environment, or bath \cite{Breuer:book}. It has largely been built on a framework of completely positive (CP) maps \cite{Stinespring:1955,Kraus:book}.  Complete positivity \cite{comment:CPDef} is implied by the assumption that the initial system-bath states are uncorrelated (i.e., tensor product states), but it can arise even when the initial system-bath states are classically correlated \cite{Rod:08,SL}. Typically, CP maps are good models of subsystem dynamics when system-bath correlations (or the observer's knowledge of the correlations) decays quickly compared to the time scales of the system evolution, i.e., when the system is Markovian.  Correlated initial states, and therefore non-CP maps, may be expected to be needed as the methods of quantum information theory are applied to systems as diverse as excitons in the condensed phase \cite{Rebentrost:2009}, nuclear spins in semiconductors \cite{Witzel:2006}, and quantum optical systems \cite{Liu:2011}, all exhibiting non-Markovian behavior.  When such open system evolutions can be modelled by linear subsystem dynamical maps, these maps will not typically be CP.  Despite much recent attention and progress (e.g., \cite{Stelmachovic:2001,Salgado:2002,Hayashi:2003, Jordan:2004,Salgado:2004,Shaji:2005,Carteret:2008,SL2,Rodriguez-Rosario:2010,UshaDevi:2011,Modi:2012,Brodutch:2013,McCracken:2013, McCracken:2013a,Buscemi:2014,Liu:2014,Dominy:2015}), the situation for more general initial conditions (e.g., for families of thermal states) is not well understood. Is completely positivity indispensable? Is there a consistent framework for the dynamics arising from general initial conditions? 

In this work, using a minimal set of postulates, we outline a general formulation for linear subsystem dynamical maps in the presence of initial system-bath correlations, and explore some of its implications.  In particular, we prove a representation theorem for subsystem maps giving necessary and sufficient conditions for such maps to be derivable as subsystem dynamics within the formalism.  Via this theorem we have two means of studying the relationship between initial states and subsystem evolutions -- by investigating the properties of maps arising from certain spaces of initial states, and by looking for representations for maps exhibiting certain properties.  We also demonstrate, through the construction of illustrative examples, that the resulting subsystem dynamical maps can violate fundamental theorems in quantum information theory that hold true for CP maps. Clearly, this has potentially significant implications, depending on the degree to which one is comfortable with accepting the approximation that is invariably made when assuming uncorrelated initial conditions. If this approximation is not made, then as we shall demonstrate, various widely accepted results in the theory of open quantum systems must be revisited.

\textit{Why CP maps?}---%
Completely positive maps are ubiquitous in the theory of open quantum systems \cite{Breuer:book} and in quantum information theory \cite{Nielsen:book,Wilde:book} because the structure of CP maps makes them conducive to the development of a mathematical theory of such systems \cite{Kraus:1971,Gorini:1976uq,Lindblad:1976}, and because the dynamics of many experimentally relevant open quantum systems are well-approximated by 
master equations which integrate to CP maps \cite{Alicki:87}.  The existence of non-CP dynamics is implicit in recent work to measure non-Markovianity of open quantum systems, which is commonly done by quantifying the failure of subsystem dynamics to be completely positive \cite{Breuer:2009, Laine:2010, Breuer:2012, Rodriguez-Rosario:2012, Bylicka:2014}.  However, even in this literature, it is not uncommon to see non-CP maps described as less ``physical'' than CP maps \cite{Wolf:2008}.  We attempt to address this discrimination below, showing that non-CP maps arise within realistic conditions. 
A common argument (e.g., \cite{Lindblad:1976,Preskill:1998,Nielsen:book}) made to justify CP maps is that, independent of any assumptions on the system and bath, the reduced dynamics of the system must be CP because of possible entanglement between the system and a third ``witness'' system which is a closed system with zero Hamiltonian (the witness is ``dead'' \cite{Pechukas:1994}).  The claim is that, since such witnesses may exist, and since the joint system-witness evolution is taken to be $\Psi\otimes\id$, where $\Psi$ is the dynamical map of the system, $\Psi\otimes\id$ must be positive so that any entangled system-witness state is evolved correctly.  Since this must hold for all possible witnesses, $\Psi$ must be CP.  After we have established some basic definitions we will argue that this 
notion of a physical mandate for complete positivity is poorly motivated. To complete the picture, it is necessary to consider the role of non-CP (even non-positive) maps in describing the dynamics of some types of open systems.  Indeed, as we will show, essentially any $\mathbb{C}$-linear, Hermiticity preserving, trace preserving 
map has a role to play in describing subsystem dynamics.

\textit{A formalism for subsystem dynamics}.---%
\label{sec:FormalismSummary}%
We begin by summarizing the formalism which is developed in more detail in \cite{Dominy:2015}.  Fix finite dimensional Hilbert spaces $\HS$ and $\HB$ for the system and bath, respectively.  The standard (Kraus \cite{Kraus:1971}) construction of subsystem dynamical maps can be described by fixing a bath state $\rhoB$ and considering the operator subspace $\CS = \BHS\otimes \rhoB = \{A\otimes\rhoB : A\in\BHS\}\subset \BHSB$, where $\mathcal{B}(\mathcal{H})$ denotes the space of bounded operators acting on the Hilbert space $\mathcal{H}$.  The evolution of an initial state of the subsystem $\rhoS$ under the action of a unitary operator $U\in\USB$ is then uniquely defined by $\Psi_{U}(\rhoS) = \TrB[\Ad_{U}(\rhoS\otimes\rhoB)]$, where $\Ad_{U}(X) \equiv UXU^\dagger$.  It may be seen that this $\Psi_{U}$ is the unique map making the diagram in Figure \ref{fig:KrausCD} commute.  

In order to move from this construction to one in which arbitrary initial system-bath correlations may be modeled, we wish to generalize from spaces of the form $\BHS\otimes\rho_{B}$ to all subspaces $\CS\subset\BHSB$ that give rise to unique subsystem maps $\Psi_{U}^{\CS}:\TrB\CS\mapsto\BHS$.  To that end, let $\mathcal{D}_{\textsc{q}}$ ($\textrm{Q} \in \{\textrm{S},\textrm{B},\textrm{SB}\}$) denote the convex set of all density matrices (i.e., positive, unit trace operators) of system $\textrm{Q}$ and define $U$-consistent subspaces as follows:
\begin{mydefinition}
\label{def:UCS}
For a fixed $U\in\USB$, we say that $\CS\subset\BHSB$ is a \emph{$U$-consistent subspace} if
(i) $\CS$ is a $\mathbb{C}$-linear subspace,
(ii) $\CS$ is spanned by states, i.e., $\Span_{\mathbb{C}}[\DSB\cap\CS]=\CS$, and
(iii) if $X,Y\in\CS$ are such that $\TrB X = \TrB Y$ then $\TrB[\Ad_{U}(X)] = \TrB[\Ad_{U}(Y)]$. 
If $\Sgp\subset\USB$ and $\CS$ is $U$-consistent for all $U\in\Sgp$, then we say that $\CS$ is a $\Sgp$-consistent subspace.
	\end{mydefinition}
This definition comprises essentially a minimal set of assumptions for obtaining a unique $\mathbb{C}$-linear \cite{comment:nonlinear} subsystem dynamical map $\Psi_{U}^{\CS}:\TrB\CS\mapsto\BHS$.  This means that there is a \emph{unique} $\Psi_{U}^{\CS}$ such that $\Psi_{U}^{\CS}[\TrB(A)] =\TrB[\Ad_{U}(A)]$  $\forall A\in\CS$; in other words, there is a unique $\Psi_{U}^{\CS}$ that makes the diagram in Figure~\ref{fig:commutativeDiagram} commute.  By construction, that map $\Psi_{U}^{\CS}$ is $\mathbb{C}$-linear, Hermiticity-preserving, and trace-preserving.  Because of point (ii) in the definition, $\Psi_{U}^{\CS}$ is also uniquely defined by the unitary evolution of the density matrices in $\CS$, i.e., by all of the physically meaningful evolutions allowed by the choices of $U$ and $\CS$.  The assumptions of Definition \ref{def:UCS} are weaker than those implicit in the standard linear assignment map formalism \cite{Pechukas:1994, Alicki:1995, Pechukas:1995}.  However, it should be stressed that, within this more general framework, there is a trade-off between the set of admissible initial states $\CS\cap\DSB$ and the set of unitary operators $U$ for which $\CS$ is $U$-consistent.  Typically, the larger the set $\Sgp\subset\USB$, the more restrictive the conditions for $\Sgp$-consistency, with the linear assignment map case representing the limit $\Sgp = \USB$.  Of course, other choices of the set of initial states may give rise to other scenarios \cite{Dominy:2015}.  In particular, if the set of initial states is not $U$-consistent then there is no corresponding subsystem dynamical map, and if the set of initial states is $U$-consistent, but not convex, the subsystem dynamical map may fail to be linear \cite{Romero:2004}.  Nevertheless, we focus in this paper on the situations giving rise to linear reduced dynamics, in order to better understand these cases.

The greater generality of this approach over that of assignment maps may be useful for modeling the evolution of realistic initial states like families of thermal states, as we demonstrate in the following example that manifestly falls outside the scope of the linear assignment map formalism.

\begin{myexample}
\label{ex:1}
	Consider a one-qubit system in contact with a one-qubit bath, with the following parametrized system-bath Hamiltonian
	\begin{equation}
		H(\theta) = \theta(X+Z)\otimes \identity + X\otimes X,
	\end{equation}
As shown in the Supplementary Material, 
the Gibbs states $\rho = e^{-\beta H(\theta)}/Z_{\textsc{sb}}$ of this family of Hamiltonians span the space of \emph{correlated} initial states
	\begin{equation}
		\CS = \Span_{\mathbb{C}}\{\identity, X\otimes \identity, Z\otimes\identity, \identity\otimes X, X\otimes X, Z\otimes X\},
	\end{equation}
	which is $U$-consistent when, for example, $U$ is generated by the controlled phase operator, i.e., $U = \exp(-it K)$, where $K = (\identity + Z\otimes \identity + \identity\otimes Z - Z\otimes Z)/2$.  The resulting subsystem dynamical map $\Psi_{U}^{\CS}(A) = \cos^{2}(t)A + \sin^{2}(t)(A+ZAZ)/2 + i\sin(t)\cos(t)[A,Z]/2$ is completely positive for all $t\geq 0$.  Indeed, $\Psi_{U}^{\CS}(A) = E_{1}AE_{1}^{\dag} + E_{2}AE_{2}^{\dag}$, where $E_{1} = \sqrt{(1+\cos(t))/2}\big(\cos(t/2)\identity -i\sin(t/2)Z)$, and $E_{2} = \sqrt{(1-\cos(t))/2}\big(\sin(t/2)\identity -i\cos(t/2)Z)$. Thus a family of correlated, thermal initial states, along with a family of unitary transformations, can give rise to CP subsystem dynamics.
\end{myexample}
\begin{figure}
	\subfigure[]{\includegraphics[width=0.54\columnwidth]{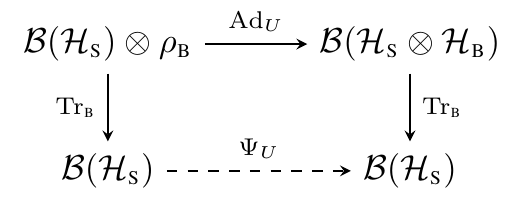}\label{fig:KrausCD}}
	\subfigure[]{\includegraphics[width=0.45\columnwidth]{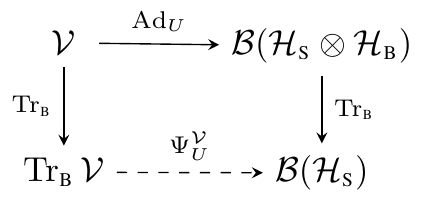}\label{fig:commutativeDiagram}}
	\caption{(a) This commutative diagram uniquely defines the subsystem dynamical map $\Psi_{U}^{\CS}$ arising in the Kraus formulation \cite{Kraus:1971} from the system-bath evolution operator $U\in\USB$. (b) For any unitary evolution operator $U\in\USB$ and any $U$-consistent subspace $\CS\subset\BHSB$, this commutative diagram uniquely defines the $\mathbb{C}$-linear, Hermiticity-preserving, trace-preserving map $\Psi_{U}^{\CS}$ which acts as the time evolution operator for system states in $\TrB(\DSB\cap\CS)\subset\DS$.}
\end{figure}

\textit{Witnessed (Complete) Positivity}.---%
We can now address the supposed physical mandate for complete positivity suggested above.  It is often argued that \cite{Lindblad:1976,Preskill:1998,Nielsen:book}, in order that a subsystem dynamical map $\Psi_{U}^{\CS}:\TrB\CS\mapsto\BHS$ be considered ``physical'', $\Psi_{U}^{\CS}\otimes\id_{\BHW}$ must be positive for all finite dimensional closed ``witness" systems $\HW$, so that if the system is initially correlated with the witness, the evolution $\Psi_{U}^{\CS}\otimes\Ad_{W}$ always yields a state in $\DSW$ for any $W\in\UW$.
Indeed, if $\Psi_{U}^{\CS}$ is not CP, there exist witnesses $\HW$ and system-witness states $\rho_{\textsc{sw}}\in\DSW\cap[\TrB\CS\otimes\BHW]$ such that $\Psi_{U}^{\CS}\otimes\id(\rho_{\textsc{sw}})\not\geq 0$.  Since this type of evolution is nonsense (a state evolving to a non-state) such non-CP maps are often dismissed as non-physical.

However, note the choice of a $U$-consistent subspace $\CS$ may be thought of as a ``promise'' that the resulting subsystem dynamical map $\Psi_{U}^{\CS}$ will only be applied to evolve the reduced state of the system when the initial system-bath state lies in $\DSB\cap \CS$ \cite{comment:why-in-DSBV}.  
This implies other promises.  First, since the initial system-bath state lies in $\DSB\cap\CS$, it follows that the initial system state lies in the ``physical domain" $\TrB[\DSB\cap\CS]\subset\DS\cap\TrB\CS$.  
Second, as we show in the Supplementary Material, when other witness systems beyond the system and bath are considered, the total system-bath-witness state must in many cases lie in $\DSBW\cap[\CS\otimes\BHW]$ if the system-witness evolution is to be described by $\Psi_{U}^{\CS}\otimes\id_{\BHW}$ (see Supplementary Material for more details).  
As a consequence, the only system-witness states admissible are those in the ``witnessed physical domain'' $\TrB[\DSBW\cap[\CS\otimes\BHW]]$, which is a \emph{subset} of $\DSW\cap[\TrB\CS\otimes\BHW]$.  Since any system-witness state $\rho_{\textsc{sw}}$ in this domain is covered by a system-bath-witness state $\rho_{\textsc{sbw}}\in\DSBW\cap[\CS\otimes\BHW]$, it holds that $\Psi_{U}^{\CS}\otimes\id_{\BHW}(\rho_{\textsc{sw}}) = \TrB[\Ad_{U\otimes\identity}(\rho_{\textsc{sbw}})] \in\DSW$.  
Other initial system-witness states in $\DSW\cap [\TrB\CS\otimes\BHW]$ not belonging to the witnessed physical domain may be mapped by $\Psi_{U}^{\CS}\otimes\id_{\BHW}$ to to non-positive operators (or to states in $\DSW$), \emph{but since these initial system-witness states are never realized if the promise is upheld, the action of $\Psi_{U}^{\CS}\otimes\id_{\BHW}$ on them carries no physical meaning}.  Thus, non-CP maps should in no way be considered ``non-physical'' \cite{Shaji:2005}.

\textit{Representations of Subsystem Maps}.---%
We proceed to investigate the set of subsystem dynamical maps that can be realized within the $\Sgp$-consistent subspace framework. The following definition associates to a dynamical map a physical process that could have generated it.

\begin{mydefinition}
	Fix a finite-dimensional Hilbert space $\HS$, let $\mathcal{R}\subset \BHS$ be a $\mathbb{C}$-linear subspace spanned by states, and let $\Phi:\mathcal{R}\mapsto\BHS$.  A \emph{subsystem dynamical representation} for $\Phi$ is a triple $(\HB, U, \CS)$ such that $\HB$ is a Hilbert space, $U\in\USB$, $\CS\in\BHSB$ is a $U$-consistent subspace, $\TrB\CS = \mathcal{R}$, and $\Psi_{U}^{\CS} = \Phi$.  We say a map $\Phi:\mathcal{R}\mapsto\BHS$ is \emph{representable as subsystem dynamics} if there exists a subsystem dynamical representation $(\HB,U,\CS)$ for $\Phi$.  The \emph{positive domain}, $\Omega_{\Phi}$, of a map $\Phi:\mathcal{R}\mapsto\BHS$ is the set $\Omega_{\Phi} := \mathcal{R}\cap\DS\cap\Phi^{-1}(\DS) = \{\rho\in\mathcal{R}\cap\DS\;:\;\Phi(\rho) \in\DS\}$ of states in $\mathcal{R}$ that are mapped by $\Phi$ to states.
\end{mydefinition}

Next we give a necessary and sufficient condition for such  physical processes to exist.

\begin{mytheorem}
	\label{thm:representation}
	Let $\mathcal{R}\subset \BHS$ be a $\mathbb{C}$-linear subspace spanned by states.  A map $\Phi:\mathcal{R}\mapsto\BHS$ is representable as subsystem dynamics if and only if it is $\mathbb{C}$-linear, Hermiticity-preserving, trace-preserving, and is such that the positive domain $\Omega_{\Phi}:=\mathcal{R}\cap\DS\cap\Phi^{-1}(\DS)$ spans $\mathcal{R}$.  One representation for such a map is given by $\HB\simeq  \HS$, $\CS=\Span_{\mathbb{C}}\{\rho\otimes\Phi(\rho)\;:\;\rho\in\Omega_{\Phi}\}$, and $U = \textsc{swap}$.  
\end{mytheorem}
The proof of necessity is essentially by Definition~\ref{def:UCS}. The proof of sufficiency is a matter of checking the correctness of the given representation and is given in the Supplementary Material, along with the proofs of most of our subsequent results. Theorem \ref{thm:representation} may be thought of as essentially a unitary dilation theorem \cite{Takai:1972}, in some ways similar to those of Stinespring \cite{Stinespring:1955} and Sz.-Nagy \cite{Schaffer:1955}. Those theorems described dilations of (unital or trace-preserving) CP maps on $C^{\ast}$-algebras and of contractions on Hilbert spaces, showing that they can be represented in terms of unitary evolutions on larger spaces. Theorem~\ref{thm:representation} provides an analogous result for a more general class of subsystem maps, although it does not exhibit some of the algebraic and/or spectral features of those dilations.  This is also conceptually related to efforts to embed non-Markovian dynamics within larger systems exhibiting Markovian dynamics \cite{Breuer:1999, Breuer:2004, Budini:2013, Hush:2015}.

It should be stressed that any map $\Phi:\mathcal{R}\mapsto\BHS$ which is subsystem representable admits infinitely many representations with potentially quite diverse properties.  In particular, the use of $\textsc{swap}$ in the proof of Theorem 1 should not be taken to imply that $\textsc{swap}$ plays a central role in the formalism of $\Sgp$-consistent subspaces.  It is used here merely to demonstate one mathematically convenient and simple representation for a representable map.  The set of $\textsc{swap}$-based representations for a representable map is a vanishingly small subset of the class of all representations, and one that has no special physical significance.  The following remark gives a constructive way to describe the particular $\CS$ mentioned in Theorem \ref{thm:representation}.

\begin{myremark}
\label{rem:repStruct}
The $\CS$ described in Theorem \ref{thm:representation} may be thought of as the image through $\id\otimes\Phi$ of the symmetric sector (i.e., the $+1$ eigenspace of the $\Ad_{\textsc{swap}}$ operator) of $\mathcal{R}\otimes\mathcal{R}$, i.e., $\CS = \id\otimes\Phi\big(\Span_{\mathbb{C}}\{A_{i}\otimes A_{j} + A_{j}\otimes A_{i}\;:\; i\leq j\}\big)$ where $\{A_{i}\}$ is any basis for $\mathcal{R}$.  For this representation, the physical domain $\TrB[\CS\cap\DSB]$ is equal to the positive domain $\Omega_{\Phi}$ of $\Phi$.
\end{myremark}

Let us take a closer look at the construction described in Theorem \ref{thm:representation} by applying it to obtain a representation of the transpose map on a single qubit -- a map well-known to be positive, but not completely positive. 
	 Let $\HS\simeq \HB\simeq\mathbb{C}^{2}$, $U=\textsc{swap}$, and
	\begin{equation}
		\CS := \Span_{\mathbb{C}}\{\rho\jmdotimes\rho^{T}\;:\;\rho\in\DS\},
	\end{equation}
yielding $\Psi_{U}^{\CS}(\rho) = \rho^{\rmT}$. Using Remark~\ref{rem:repStruct}, it may be seen that $\CS$ is a $10$-dimensional operator space (a basis for this space is found in the SM).  This example demonstrates the added generality of the $\Sgp$-consistent subspace formalism to that of Pechukas-like assignment maps.  Whereas this construction has produced a not-completely-positive map on the full domain $\BHS$ with physical domain equal to $\DS$, the Pechukas theorem \cite{Pechukas:1994} shows that this cannot be reproduced with any standard linear assignment map $\mathcal{A}:\TrB\CS \mapsto\CS$.

	It may be tempting to view the definition of the $U$-consistent subspace $\CS:=\Span_{\mathbb{C}}\{\rho\otimes\Phi(\rho)\;:\; \rho\in\Omega_{\Phi}\}$ in Theorem \ref{thm:representation} as in some way nonlinear, perhaps invoking the nonlinear map $\rho\mapsto \rho\otimes\Phi(\rho)$, or indeed, the cloning map $\rho\mapsto\rho\otimes\rho$.  However, remember that the specification of $\CS$ is really a ``promise'' that, when the subsystem dynamical map $\Psi_{U}^{\CS}$ will be invoked, the initial system-bath state lies in $\DSB\cap \CS$.  There is no assumption about how the system-bath state came to be in $\DSB\cap\CS$ or, in the case of the construction described in Theorem \ref{thm:representation}, about what the map $\rho\mapsto\rho\otimes\Phi(\rho)$ means operationally.  We seek only to answer the question: if the initial system-bath state is in $\DSB\cap\CS$, how does the reduced state of the system evolve from there?  That question involves only linear maps; in particular, every space and every map in Figure \ref{fig:commutativeDiagram} is $\mathbb{C}$-linear.

Theorem \ref{thm:representation} says that the representability of a map $\Phi$ depends, in part, on the physical domain $\Omega_{\Phi}$ spanning $\mathcal{R}$.   We can weaken this condition by discarding those parts of $\Phi$ that play no role in mapping positive operators to positive operators, i.e., that are irrelevant for describing physical evolutions:
\begin{mycorollary}
	If $\Phi:\BHS\mapsto\BHS$ is $\mathbb{C}$-linear, Hermiticity-preserving, and trace-preserving with non-empty positive domain $\Omega_{\Phi} = \DS\cap\Phi^{-1}(\DS)$, then the restriction $\Phi\big|_{\mathcal{R}}:\mathcal{R}\mapsto\BHS$ is representable as subsystem dynamics, where $\mathcal{R} = \Span_{\mathbb{C}}\Omega_{\Phi}$.  In other words, the physically relevant part of $\Phi$ is representable as subsystem dynamics.
\end{mycorollary}

\textit{Violations of QIT Theorems}.---%
\label{sec:Violations}%
Much of quantum information theory has been developed under the assumption that subsystem dynamical maps are trace-preserving CP (CPTP) maps with domain the full system operator algebra $\BHS$. What happens when this assumption is violated? We construct a simple example of a non-CP dynamical map $\Phi$, representable as subsystem dynamics, that violates some cherished inequalities for CPTP maps that are fundamental to quantum information theory.  

The theorems we will consider -- contractivity of CPTP maps under the trace norm and Uhlmann's theorem of non-increasing relative entropy
-- both serve to characterize the same basic phenomenon: output states of CPTP maps are less distinguishable than the input states.  That these theorems are violable is implicit in the definitions of several measures of non-Markovianity \cite{Breuer:2009, Laine:2010, Bylicka:2014}, however the violations below are based on the explicit unitary evolution of a particular subspace of the system-bath operator algebra so as to have a clear physical basis.  To construct our example, therefore, we must define a dynamical map for which some output states are more distinguishable than the inputs (see also \cite{Breuer:2009,Dajka:2010,Dajka:2012}).  To that end, let the system of interest be a single qubit and consider the $\mathbb{C}$-linear, Hermiticity-preserving, trace-preserving map $\Phi:\BHS\mapsto\BHS$ given by 
\begin{equation}
	\Phi(A) = \frac{1}{\epsilon}A - \frac{1-\epsilon}{2\epsilon}\Tr(A)\identity
\label{eq:repo-map}
\end{equation}
for some small $\epsilon>0$.  This $\Phi$ may be thought of as a ``repolarizer''; it is the inverse of the depolarizing channel $\rho\mapsto \epsilon\rho + (1-\epsilon)\identity/2$.  In contrast to the example of the transpose map which is positive but not CP, this $\Phi$ is not even positive.  The positive domain of $\Phi$ may be seen to be the $\epsilon$-ball 
\begin{equation}
	\Omega_{\Phi} = \left\{\frac{1}{2}\big(\identity + a_{1}X + a_{2}Y + a_{3}Z\big)\;:\; \|\vec{a}\|^{2} \leq \epsilon^{2}\right\}.
\end{equation}

We may construct a subsystem dynamical representation $(\HB, U, \CS)$ as described in Theorem~\ref{thm:representation} as follows. Let $\HB$ be another qubit, $U = \textsc{swap}$, and
$\CS = \Span_{\mathbb{C}}\{\rho\otimes\Phi(\rho)\;:\;\rho\in\Omega_{\Phi}\}$.
Invoking Remark~\ref{rem:repStruct} readily yields the characterization
of $\CS$ as a 10-dimensional subspace of $\BHSB$ (specified in the SM).  It follows from Theorem \ref{thm:representation} that the resulting subsystem dynamical map $\Psi_{U}^{\CS}$ is identical to $\Phi$, and the physical domain is identical to $\Omega_{\Phi}$. \emph{Thus, the repolarizer map is not fiction: it corresponds to a well-defined physical process.}

Alternatively, we may construct a different subsystem dynamical representation for $\Phi$ as follows.  If a map $\Phi:\BHS\mapsto\BHS$ represented by $(\HB,U,\CS)$ is invertible, the inverse map $\Phi^{-1}:\BHS\mapsto\BHS$ may be represented by $\big(\HB, U^{\dag}, \Ad_{U}\CS\big)$.  The physical domain of this new representation is the image $\Phi(\TrB[\CS\cap\DSB])\subset\DS$ of the physical domain $\TrB[\CS\cap\DSB]$ of the representation $(\HB,U,\CS)$ of $\Phi$.  In the present case, we may write the depolarizing channel $\Phi^{-1}$ as 
\begin{align}
	\Phi^{-1}(A) 
	 = \epsilon A + \frac{1-\epsilon}{4}\sum_{i=0}^{3}\sigma_{i}A\sigma_{i}
	= \sum_{i=0}^{3}M_{i}AM_{i}^{\dag},
\end{align}
where $M_{0} = \frac{\sqrt{1+3\epsilon}}{2}\identity$ and $M_{i} = \frac{\sqrt{1-\epsilon}}{2}\sigma_{i}$ for $i=1,2,3$.  Then $\HB \simeq\mathbb{C}^{4}\simeq \Span_{\mathbb{C}}\{\ket{0},\ket{1},\ket{2},\ket{3}\}$, 
\begin{equation}
	\CS = \bigg\{\sum_{i,j}M_{i}AM_{j}^{\dag}\otimes\ketbra{i}{j}\;:\; A\in\BHS\bigg\}
\end{equation}
and $\langle i|U|0\rangle = M_{i}$ (the rest of $U$ is arbitrary up to the constraint $U\in\USB$).


\textit{Contractivity Under Norm Distance}.---%
For any $1\leq p\leq \infty$, let $\delta_{p}$ be the normalized von Neumann-Schatten $p$-norm distance, i.e.,
\begin{equation}
	\delta_{p}(\tau_{1},\tau_{2}) = 2^{-\frac{1}{p}}\|\tau_{1} - \tau_{2}\|_{p} = 2^{-\frac{1}{p}}\big[\Tr(|\tau_{1}-\tau_{2}|^{p})\big]^{\frac{1}{p}}. 
\end{equation}
It is well known (see, e.g., \cite{Nielsen:book}) that any trace preserving CP map on a C$^{*}$ algebra is contractive under the 1-norm (i.e., the trace norm).  In other words, if $\Upsilon$ is a CPTP map with domain $\BHS$, and $\tau_{1}$ and $\tau_{2}$ are states in $\DS$, then $\delta_{1}(\tau_{1},\tau_{2})\geq \delta_{1}[\Upsilon(\tau_{1}),\Upsilon(\tau_{2})]$.  However, for the non-positive repolarizer map $\Phi$ and for any norm $\|\cdot\|$ on $\BHS$, 
\begin{align}
	\|\Phi(\rho_{1}) - \Phi(\rho_{2})\| 
	= \frac{1}{\epsilon}\|\rho_{1}-\rho_{2}\|
\end{align}
for all $\rho_{1},\rho_{2}\in\Omega_{\Phi}$.  So, the subsystem dynamical map $\Psi_{U}^{\CS}$ is not contractive for any norm distance.


\textit{Uhlmann's Theorem}.---%
Consider the relative entropy $S(\tau_{1}\|\tau_{2}) := \Tr(\tau_{1}[\log(\tau_{1}) - \log(\tau_{2})])$.  Uhlmann's theorem \cite{Lindblad:1975, Uhlmann:1977} states that the relative entropy is non-increasing for CPTP maps, i.e., $S(\tau_{1}\|\tau_{2}) \geq S[\Phi(\tau_{1})\|\Phi(\tau_{2})]$ for any CPTP map $\Phi$ and any states $\tau_{1},\tau_{2}$.  However, for the repolarizer map the joint convexity of the relative entropy \cite{Lieb:1973, Uhlmann:1977} implies that
\begin{align}
	S(\rho_{1}\|\rho_{2}) &= S\big(\epsilon\Phi(\rho_{1}) + (1-\epsilon)\identity/2 \|\epsilon\Phi(\rho_{2}) + (1-\epsilon)\identity/2\big) \nonumber\\
	& \leq \epsilon S\big(\Phi(\rho_{1})\|\Phi(\rho_{2})\big) + (1-\epsilon)S\big(\identity/2\|\identity/2\big) \nonumber\\
	& = \epsilon S\big(\Phi(\rho_{1})\|\Phi(\rho_{2})\big),
\end{align}
so that $S\big(\Phi(\rho_{1})\|\Phi(\rho_{2})\big)\geq \frac{1}{\epsilon}S(\rho_{1}\|\rho_{2})$ for all $\rho_{1},\rho_{2}\in\Omega_{\Phi}$.
Therefore the repolarizer map fails to be non-increasing for the relative entropy.


\textit{Summary and Open Questions}.---%
\label{sec:Summary}%
We have shown that, far from being limited to CP maps, essentially any $\mathbb{C}$-linear, Hermiticity-preserving
trace-preserving 
map can arise as the unique subsystem dynamical map for some bath $\HB$, unitary propagator $U\in\USB$, and space of (typically correlated) initial states $\CS\subset\BHSB$.    Non-CP maps such as the transpose map and even non-positive maps such as the repolarizer map [Eq.~\eqref{eq:repo-map}], often considered ``nonphysical'' arise easily within the formalism we have presented.  However, their application requires somewhat more care than CP maps if the sequence of ``promises'' suggested by $\CS$ are not enforced by physics (e.g., symmetries, or limiting behavior similar to weak interactions).  Likewise, concatenation of these maps must be handled with care.

The inclusion of non-CP maps in the space of possible subsystem dynamical maps represents a challenge to the standard view.  
We stress that this is 
in our view a matter of necessity if the theory of open quantum systems is to apply beyond the unrealistic assumption of classically correlated initial system-bath states. Basic, fundamental results 
may need to be reformulated and reproved if the extension is to be achieved.  An important part of this work will be to understand under what circumstances, if any, the different behavior of these non-CP maps may be exploited for quantum gains.

\acknowledgments
This research was supported by the ARO MURI grant W911NF-11-1-0268.  The authors thank Iman Marvian for many helpful discussions. 

\bibliography{refs}

\begin{thebibliography}{56}
\bibitem [{com({\natexlab{a}})}]{comment:CPDef}%
  \BibitemOpen
  \href@noop {} {} ({\natexlab{a}}),\ \bibinfo {note} {{As a reminder, if $\HH$
  and $\mathcal{K}$ are Hilbert spaces, $\mathcal{R}\subset\BH$ is a
  $\mathbb{C}$-linear subspace spanned by states, and
  $F:\mathcal{R}\to\BL(\mathcal{K})$ is $\mathbb{C}$-linear, then $F$ is
  completely positive if
  $F\otimes\id:\mathcal{R}\otimes\BHW\to\BL(\mathcal{K})\otimes\BHW$ is a
  positive map for all finite dimensional $\HW$.}}\BibitemShut {Stop}%
\end{thebibliography}

\appendix
\section{Calculations in support of Example 1 (joint initial Gibbs state)}
\label{app:Gibbs}
\setcounter{mytheorem}{0}
\setcounter{mycorollary}{0}

Consider a one-qubit system in contact with a one-qubit bath, with the following parametrized system-bath Hamiltonian
\begin{equation}
	H(\theta) = \theta(X+Z)\otimes \identity + X\otimes X.
\end{equation}
This gives rise to the thermal states
\begin{subequations}
\begin{align}
	& \rho(\theta,\beta) = e^{-\beta H(\theta)}/\Tr(e^{-\beta H(\theta)})\\
	& = \frac{1}{4}\Bigg\{\identity - \frac{\theta}{\lambda\gamma}\left[\frac{\gamma\sinh(\beta\lambda) + \lambda\sinh(\beta\gamma)}{\cosh(\beta\lambda) + \cosh(\beta\gamma)}\right]Z\otimes \identity \nonumber\\
		&\   - \frac{1}{\lambda\gamma}\left[\frac{(\theta+1)\gamma\sinh(\beta\lambda) + (\theta-1)\lambda\sinh(\beta\gamma)}{\cosh(\beta\lambda) + \cosh(\beta\gamma)} \right]X\otimes \identity \nonumber\\
		& \  + \left[\frac{\cosh(\beta\lambda) - \cosh(\beta\gamma)}{\cosh(\beta\lambda) + \cosh(\beta\gamma)}\right]\identity\otimes X \nonumber\\
		& \ -  \frac{1}{\lambda\gamma}\left[\frac{(\theta+1)\gamma\sinh(\beta\lambda) - (\theta-1)\lambda\sinh(\beta\gamma)}{\cosh(\beta\lambda) + \cosh(\beta\gamma)} \right]X\otimes X \nonumber\\
		& \ -\frac{\theta}{\lambda\gamma}\left[\frac{\gamma\sinh(\beta\lambda) - \lambda\sinh(\beta\gamma)}{\cosh(\beta\lambda) + \cosh(\beta\gamma)}\right]Z\otimes X\Bigg\},
\end{align}
\end{subequations}
where
\begin{subequations}
\label{eq:lambda-gamma}
\begin{align}
	\lambda & = \sqrt{2\theta^{2} + 2\theta + 1}\\
	\gamma & = \sqrt{2\theta^{2} - 2\theta + 1}.
\end{align}
\end{subequations}

\begin{proof}
To compute these states, note first that
\begin{subequations}
\begin{align}
	H^{2} & = \theta^{2}[X^{2} + \{X,Z\} + Z^{2}]\otimes \identity \nonumber\\
	& \qquad + \theta[2X^{2} + \{X,Z\}]\otimes X + X^{2}\otimes X^{2}\\
	& = (2\theta^{2}+1)\identity + 2\theta (\identity\otimes X).
\end{align}
\end{subequations}
It follows that $H^{2k}$ has the form
\begin{equation}
	H^{2k} = a_{k}\identity + b_{k} \identity\otimes X
\end{equation}
where $a_{k}$ and $b_{k}$ satisfy the recurrence relation
\begin{align}
	\begin{bmatrix}a_{k+1}\\b_{k+1}\end{bmatrix} & = \begin{bmatrix}2\theta^{2}+1  & 2\theta\\2\theta & 2\theta^{2}+1\end{bmatrix}\begin{bmatrix} a_{k}\\ b_{k}\end{bmatrix} & \begin{bmatrix}a_{0}\\ b_{0}\end{bmatrix} & = \begin{bmatrix}1\\ 0\end{bmatrix}
\end{align}
which is readily solved to obtain
\begin{equation}
	\begin{bmatrix}a_{k}\\ b_{k}\end{bmatrix} = \frac{1}{2}\begin{bmatrix}(2\theta^{2} + 2\theta + 1)^{k} + (2\theta^{2} - 2\theta + 1)^{k}\\ (2\theta^{2} + 2\theta + 1)^{k} - (2\theta^{2} - 2\theta + 1)^{k}\end{bmatrix}.
\end{equation}
It is then found that the even and odd powers of the Hamiltonian are given by
\begin{subequations}
\begin{align}
	H^{2k} & = \frac{1}{2}[(2\theta^{2} + 2\theta + 1)^{k} + (2\theta^{2} - 2\theta + 1)^{k}]\identity \nonumber\\
	& \quad + \frac{1}{2} [(2\theta^{2} + 2\theta + 1)^{k} - (2\theta^{2} - 2\theta + 1)^{k}]\identity\otimes X\\
	H^{2k+1} & = \frac{\theta}{2}[(2\theta^{2} + 2\theta + 1)^{k} + (2\theta^{2} - 2\theta + 1)^{k}]Z\otimes \identity \nonumber\\
	& \quad + \frac{1}{2} \Big[(\theta+1)(2\theta^{2} + 2\theta + 1)^{k} \nonumber\\
	& \quad\qquad + (\theta-1)(2\theta^{2} - 2\theta + 1)^{k}\Big]X\otimes \identity\nonumber\\
	& \quad + \frac{1}{2}\Big[(\theta+1)(2\theta^{2} + 2\theta + 1)^{k} \nonumber\\
	& \quad\qquad - (\theta-1)(2\theta^{2} - 2\theta + 1)^{k}\Big]X\otimes X\nonumber\\
	& \quad + \frac{\theta}{2} [(2\theta^{2} + 2\theta + 1)^{k} - (2\theta^{2} - 2\theta + 1)^{k}]Z\otimes X,
\end{align}
\end{subequations}
whence
\begin{align}
	& 2e^{-\beta H(\theta)}\nonumber\\
	& = [\cosh(\beta\lambda) + \cosh(\beta\gamma)]\identity \nonumber\\
	& \quad  + [\cosh(\beta\lambda) - \cosh(\beta\gamma)]\identity\otimes X \nonumber\\
	& \quad  - \left[\theta\frac{\sinh(\beta\lambda)}{\lambda} + \theta\frac{\sinh(\beta\gamma)}{\gamma}\right]Z\otimes \identity \nonumber\\
	& \quad - \left[(\theta+1)\frac{\sinh(\beta\lambda)}{\lambda} + (\theta-1)\frac{\sinh(\beta\gamma)}{\gamma}\right]X\otimes \identity \nonumber\\
	& \quad - \left[(\theta+1)\frac{\sinh(\beta\lambda)}{\lambda} - (\theta-1)\frac{\sinh(\beta\gamma)}{\gamma}\right]X\otimes X \nonumber\\
	& \quad - \left[\theta\frac{\sinh(\beta\lambda)}{\lambda} - \theta\frac{\sinh(\beta\gamma)}{\gamma}\right]Z\otimes X,
\end{align}
where $\lambda$ and $\gamma$ are given in Eq.~\eqref{eq:lambda-gamma}.
\end{proof}

Since the coefficient functions are linearly independent, we find, then, that
\begin{subequations}
\begin{align}
	\CS & := \Span_{\mathbb{C}}\{\rho(\theta,\beta)\;:\;\theta\in\mathbb{R}\}\nonumber\\
	& = \Span_{\mathbb{C}}\{\identity,\identity\otimes X, Z\otimes\identity, X\otimes\identity, X\otimes X, Z\otimes X\} \nonumber\\
	& = \Span_{\mathbb{C}}\{\identity,X,Z\}\otimes\Span_{\mathbb{C}}\{\identity,X\}\\
	\TrB\CS & = \Span_{\mathbb{C}}\{\identity,X, Z\}\\
	\CS_{0} & = \Span_{\mathbb{C}}\{\identity\otimes X, X\otimes X, Z\otimes X\},
\end{align}
\end{subequations}
so that the subsystem dynamical map $\Psi_{U}^{\CS}$ exists for all $U\in\USB$ such that $\TrB[\Ad_{U}(\identity\otimes X)] = \TrB[\Ad_{U}(X\otimes X)] = \TrB[\Ad_{U}(Z\otimes X)] = 0$.

Now take $U = \exp(-itK)=\cos(t)\identity - i\sin(t)K$, where $K = \frac{1}{2}(\identity + \sigma\otimes\identity + \identity\otimes \tau - \sigma\otimes\tau)$, and where $\sigma$ and $\tau$ are Pauli operators, so that $K$ can be thought of as a generalized controlled-Pauli operator.  Suppose further that $\{\tau,X\} = 0$.  Then 
\begin{align}
	\Ad_{U}(\Gamma\otimes X)&=\cos^{2}(t)\Gamma\otimes X + \frac{1}{4}\sin^{2}(t)\big(2\{\sigma,\Gamma\}\otimes X\nonumber\\
	& \quad  + [\Gamma,\sigma]\otimes[\tau,X]\big)\nonumber\\
	&  +\frac{i}{4}\sin(t)\cos(t)\big(2[\Gamma,\sigma]\otimes X + 2\Gamma\otimes[X,\tau]\nonumber\\
	&  \quad - \{\Gamma,\sigma\}\otimes[X,\tau]\big)
\end{align}
so that $\TrB[\Ad_{U}(\Gamma\otimes X)] = 0$ for all $\Gamma = \identity, X, Z$.  So $\CS$ is $U$-consistent with respect to these $U$ operators.  Then
\begin{align}
	\Ad_{U}(\Gamma\otimes \identity)&=\cos^{2}(t)\Gamma\otimes\identity + \frac{\sin^{2}(t)}{2}\big((\Gamma+\sigma\Gamma\sigma)\otimes\identity\nonumber\\
	& \quad + (\Gamma-\sigma\Gamma\sigma)\otimes\tau\big) \nonumber\\
	&+ \frac{i}{2}\sin(t)\cos(t)[\Gamma,\sigma]\otimes(\identity-\tau),
\end{align}
so that
\begin{align}
	\Psi_{U}^{\CS}(A) &= \TrB[\Ad_{U}(A\otimes\identity/2)] \nonumber\\
	&=\cos^{2}(t)A + \frac{\sin^{2}(t)}{2}\big(A+\sigma A\sigma\big)\nonumber\\
	&+ \frac{i}{2}\sin(t)\cos(t)[A,\sigma].
\end{align}
for all $A\in\Span_{\mathbb{C}}(\identity, X, Z)$.

\section{Maximal $\Sgp$-consistent subspaces}%
\label{sec:maxSubspaces}%
Here we show that the witness consistent subspace $\CS\otimes\BHW$ are not simply a convenient choice, but in many cases these subspaces are the only choice for initial system-bath-witness states if the system-witness evolution is to be described by $\Psi_{U}^{\CS}\otimes\id$.  However, the assumption that system-witness states evolve as $\Psi_{U}^{\CS}\otimes\Ad_{W}$ (rather than some joint system-witness evolution that doesn't factor as a tensor product of maps) and the restrictions that this places on the admissible system-bath-witness states should not be taken lightly. The fact that complete positivity implicitly imposes these restrictions should raise considerable doubt as to the physical significance of this property.

Once a subsystem dynamical map has been defined it is natural to extend its domain to be as large as possible. Formally:
\begin{mydefinition}
	Fix $\HS$, $\HB$, and $\Sgp\subset\USB$.  A $\Sgp$-consistent subspace $\CS\subset\BHSB$ will be called \emph{maximal} if for every $\rhoSB\in\DSB\setminus\CS$, the subspace $\CS+\mathbb{C}\rhoSB$ is not $\Sgp$-consistent.
\end{mydefinition}
Clearly, every $\Sgp$-consistent subspace $\CS$ may be extended to a maximal subspace, and a maximal $\Sgp$-consistent subspace $\CS$ is such that $\TrB\CS = \BHS$, i.e., \emph{all} system states are admissible.  Indeed, if $\CS$ is $\Sgp$-consistent and $\TrB\CS\neq\BHS$, then there exists $\rhoSB\in\DSB$ such that $\rhoS=\TrB\rhoSB\notin\TrB\CS$ and $\CS_{1}:=\CS+\mathbb{C}\rhoSB$ is $\Sgp$-consistent since $\ker(\TrB|_{\CS_{1}}) = \ker(\TrB|_{\CS})$, i.e. if $A,B\in\CS_{1}$ and $\TrB A = \TrB B$, then $A-B\in\CS$ so that $\TrB(\Ad_{U}(A-B)) = 0$ for all $U\in\Sgp$.  Moreover, for a maximal $\Sgp$-consistent subspace $\CS$, $\DSB\cap\CS$ is the set of all system-bath states $\rhoSB$ which transform as $\Psi_{U}^{\CS}(\TrB\rhoSB) = \TrB(\Ad_{U}\rhoSB)$ $\forall U\in\Sgp$.  In other words, the maps $\{\Psi_{U}^{\CS}\;:\;U\in\Sgp\}$ define the subspace $\CS$.  

To help us to further understand the structure of maximal subspaces consider the $\Sgp$-consistent subspace of $\ker\TrB$:
\begin{mydefinition}
	Fix $\HS$, $\HB$, and $\Sgp\subset\USB$.  The $\Sgp$-consistent subspace of $\ker\TrB\subset\BHSB$ is 
	\begin{equation}
		\hat{\CS}_{0}^{\Sgp} := \bigcap_{U\in\{\identity\}\cup\Sgp}\Ad_{U^{\dag}}\ker\TrB.
	\end{equation}
\end{mydefinition}
\begin{mylemma}
	\label{lem:transformationSpace}
	If $\CS$ is a maximal $\Sgp$-consistent subspace, then the set of all operators $A\in\BHSB$ for which $\TrB[\Ad_{U}(A)] = \Psi_{U}^{\CS}[\TrB(A)]$ for all $U\in\Sgp$ is given by $\CS' := \CS + \hat{\CS}_{0}^{\Sgp}$.
\end{mylemma}
\begin{proof}
	Since the domain of $\Psi_{U}^{\CS}$ is $\TrB\CS$, if $A\in\BHSB$ transforms as $\TrB[\Ad_{U}(A)] = \Psi_{U}^{\CS}[\TrB(A)]$, then necessarily, $\TrB(A)\in\TrB\CS$.  It follows that there exists $A'\in\CS$ such that $\TrB(A) = \TrB(A')$.  Then $0 = \Psi_{U}^{\CS}[\TrB(A-A')] = \TrB[\Ad_{U}(A-A')]$ for all $U\in\Sgp$, so that $A-A'\in\hat{\CS}_{0}^{\Sgp}$, and therefore $A\in\CS' := \CS + \hat{\CS}_{0}^{\Sgp}$.  Then all operators transforming as required lie in $\CS'$.  That all operators in $\CS'$ transform as required is trivial.
\end{proof}

As the following Lemma demonstrates, the promise that $\hat{\CS}_{0}^{\Sgp}\subset\CS$ considerably simplifies the theory:
	
\begin{mylemma}
\label{lem1}
	If $\CS$ is a maximal $\Sgp$-consistent subspace and $\hat{\CS}_{0}^{\Sgp}\subset\CS$, then for any witness $\HW$, $\CS\otimes\BHW$ is a maximal $\Sgp\otimes\identity$-consistent subspace.  In other words, $\DSBW\cap[\CS\otimes\BHW]$ is the set of all states $\rhoSBW$ in $\DSBW$ that transform as $\TrB[\Ad_{U\otimes\identity}(\rhoSBW)] = \Psi_{U}^{\CS}\otimes\id[\TrB(\rhoSBW)]$ for all $U\in\Sgp$.
\end{mylemma}

\begin{proof}
	Suppose $\rhoSBW\in\DSBW$ is such that $\TrB[\Ad_{U\otimes\identity}(\rhoSBW)] = \Psi_{U}^{\CS}\otimes\id[\TrB(\rhoSBW)]$ for all $U\in\Sgp$.  Write $\rhoSBW = \sum_{i}A_{i}\otimes W_{i}$, where $\{A_{i}\}\subset\BHSB$ and where $\{W_{i}\}\subset \BHW$ are linearly independent.  Then the condition on $\rhoSBW$ is that
	\begin{equation}
		\sum_{i}\TrB[\Ad_{U}(A_{i})]\otimes W_{i} = \sum_{i}\Psi_{U}^{\CS}[\TrB(A_{i})]\otimes W_{i}.
	\end{equation}
	By the linear independence of the $\{W_{i}\}$, this implies that $\TrB[\Ad_{U}(A_{i})] = \Psi_{U}^{\CS}[\TrB(A_{i})]$ for all $i$ and for all $U\in\Sgp$.  Then by Lemma \ref{lem:transformationSpace}, $\{A_{i}\}\subset\CS' = \CS + \hat{\CS}_{0}^{\Sgp}$, so that $\rhoSBW\in\CS'\otimes\BHW$.  Since $\hat{\CS}_{0}^{\Sgp}\subset\CS$ by assumption, it follows that $\CS' = \CS$, so that $\rhoSBW\in\CS\otimes\BHW$.  Thus $\CS\otimes\BHW$ is maximal.  And $\CS\otimes\BHW$ is a $\Sgp$-consistent subspace because it is spanned by states and for any $A\in\CS\otimes\BHW$ such that $\TrB(A) = 0$, we can again write $A = \sum_{i}A_{i}\otimes W_{i}$ with $\{A_{i}\}\subset\CS$ and the $\{W_{i}\}\subset\BHW$ linearly independent, so that $\TrB(A) = 0$ implies $\TrB(A_{i}) = 0$ for all $i$.  Then $\TrB(\Ad_{U\otimes\identity}(A)) = \sum \TrB(\Ad_{U}(A_{i}))\otimes W_{i} = \sum 0\otimes W_{i} = 0$ since $\{A_{i}\}\subset\CS$ and $\CS$ is $U$-consistent for all $U\in\Sgp$.  Therefore $\CS\otimes\BHW$ is a maximal $\Sgp\otimes\identity$-consistent subspace.
\end{proof}

Fortunately, in many cases, a maximal $\Sgp$-consistent subspace $\CS$ will indeed contain $\hat{\CS}_{0}^{\Sgp}$, e.g., when $\CS$ contains an interior state (i.e., a state with strictly positive eigenvalues), or when $\Sgp = \USB$.  In these cases $\CS' = \CS$. 

Using these tools, let us now provide an example which demonstrates---by exploiting an initial correlation between the bath and the witness---that even though the system-bath and witness subsystems evolve independently, 
the joint evolution of the system-witness is \emph{not necessarily} given by $\Psi_{U}^{\CS}\otimes\id$. 
\begin{myexample}
	Consider a Kraus subspace $\CS = \BHS\otimes\rhoB$ for some fixed  state $\rhoB\in\DB$, where $\HS\simeq \HB$.  This $\CS$ is a maximal $\USB$-consistent subspace and $\hat{\CS}_{0}^{\USB} = \{0\}\subset\CS$.  Let $\HW$ be a witness and consider a system-bath-witness state $\rhoSBW = \rhoS\otimes\rhoBW$ where $\rhoBW\subset\DBW$ is a correlated state (thus $\rhoSBW \not\in \CS \otimes \BHW$) and $\TrW\rhoBW = \rhoB$, the same bath state $\rhoB$ that defines $\CS$.  Now, let $U\in\USB$ be the $\textsc{swap}$ operator, i.e., $U\ket{\psi}\otimes\ket{\phi} = \ket{\phi}\otimes\ket{\psi}$ for all $\ket{\psi},\ket{\phi}\in\HS\simeq\HB$, so that $\Psi_{U}^{\CS}(\rhoS) = \rhoB$ for all $\rhoS\in\DS$.  Then $\rhoSBW\notin \CS\otimes\BHW$, but $\TrW \rhoSBW = \rhoS\otimes\rhoB \in\CS$, so that the evolution on the system alone is given by $\Psi_{U}^{\CS}(\Tr_{\textsc{bw}}\rhoSBW)$.  If we evolve this state $\rhoSBW$ by $U\otimes\identity$ the states of the system and bath are swapped; if we then trace over the bath we obtain an evolved system-witness state $\TrB[\Ad_{U\otimes\identity}(\rhoSBW)]  
= \rhoBW$.  However, if we trace over the bath first and then apply $\Psi_{U}^{\CS}\otimes\id$, we obtain $\Psi_{U}\otimes\id[\TrB(\rhoSBW)] = \Psi_{U}^{\CS}\otimes\id[\rhoS\otimes\rhoW] = \rhoB\otimes\rhoW$, where $\rhoW = \TrB[\rhoBW]$.  So, even though the system-bath and witness subsystems evolve independently -- each as closed systems via $\Ad_{U}$ on the system-bath and $\id$ on the witness -- the joint evolution of the system-witness is \emph{not necessarily} given by $\Psi_{U}^{\CS}\otimes\id$.
\end{myexample}

\section{Proofs}

%

\subsection{Proof of Theorem~1}
\begin{mytheorem}
	Let $\mathcal{R}\subset \BHS$ be a $\mathbb{C}$-linear subspace spanned by states.  A map $\Phi:\mathcal{R}\mapsto\BHS$ is representable as subsystem dynamics if and only if it is $\mathbb{C}$-linear, Hermiticity-preserving, trace-preserving, and is such that the positive domain $\Omega_{\Phi}:=\mathcal{R}\cap\DS\cap\Phi^{-1}(\DS)$ spans $\mathcal{R}$.  One representation for such a map is given by $\HB\simeq  \HS$, $\CS=\Span_{\mathbb{C}}\{\rho\otimes\Phi(\rho)\;:\;\rho\in\Omega_{\Phi}\}$, and $U = \textsc{swap}$.  
\end{mytheorem}
\begin{proof}
	As mentioned in the main text, any subsystem dynamical map $\Psi_{U}^{\CS}$ must be $\mathbb{C}$-linear,  Hermiticity-preserving, and trace-preserving.  It is straightforward to see that the physical domain $\TrB[\CS\cap\DSB]$ must span $\TrB\CS$ and also must lie in the positive domain of $\Psi_{U}^{\CS}$ so that the positive domain $\Omega_{\Phi_{U}^{\CS}}$ of $\Psi_{U}^{\CS}$ spans $\TrB\CS$.  Thus ``only if'' is proved.  

	Now, suppose $\Phi:\mathcal{R}\mapsto\BHS$ is $\mathbb{C}$-linear, Hermiticity-preserving, and trace-preserving, and that $\Omega_{\Phi}$ spans $\mathcal{R}$.  Then let $\HB\simeq \HS$, $U=\textsc{swap}$, and $\CS = \Span_{\mathbb{C}}\{\rho\otimes\Phi(\rho)\;:\;\rho\in\Omega_{\Phi}\}$.  It is clear from construction that $\CS$ is $\mathbb{C}$-linear and spanned by states.  If $X\in\CS$, then $X = \sum_{i}\alpha_{i}\rho_{i}\otimes\Phi(\rho_{i})$ for some $\{\alpha_{i}\}\in\mathbb{C}$ and $\{\rho_{i}\}\in\Omega_{\Phi}$.  If also $\TrB X = 0$, then $\sum_{i}\alpha_{i}\rho_{i} = 0$.  But then $\TrB (UXU^{\dag}) = \sum_{i}\alpha_{i}\Phi(\rho_{i}) = \Phi\left(\sum_{i}\alpha_{i}\rho_{i}\right) = 0$ by linearity of $\Phi$, proving that $\CS$ is a $U$-consistent subspace.  Moreover, for any $A\in\mathcal{R}$, $A = \sum_{i}\beta_{i}\sigma_{i}$ for some $\{\beta_{i}\}\subset\mathbb{C}$ and $\{\sigma_{i}\}\subset \Omega_{\Phi}$, since $\Omega_{\Phi}$ spans $\mathcal{R}$ by assumption.  Then $\CS\ni A' = \sum_{i}\beta_{i}\sigma_{i}\otimes\Phi(\sigma_{i})$ is such that $\TrB A' = A$.  It may be seen that $\TrB(U A' U^{\dag}) = \sum_{i}\beta_{i}\Phi(\sigma_{i}) = \Phi\left(\sum_{i}\beta_{i}\sigma_{i}\right) = \Phi(A)$, so that $\Psi_{U}^{\CS}(A) = \Phi(A)$ for all $A\in\mathcal{R}$.  Since $\TrB\CS = \Span_{\mathbb{C}}\Omega_{\Phi} = \mathcal{R}$, the ``if'' part of the theorem is proved.  
\end{proof}

%

\subsection{Proof of Corollary~1}
\begin{mycorollary}
	If $\Phi:\BHS\mapsto\BHS$ is $\mathbb{C}$-linear, Hermiticity-preserving, and trace-preserving with non-empty positive domain $\Omega_{\Phi} = \DS\cap\Phi^{-1}(\DS)$, then the restriction $\Phi\big|_{\mathcal{R}}:\mathcal{R}\mapsto\BHS$ is representable as subsystem dynamics, where $\mathcal{R} = \Span_{\mathbb{C}}\Omega_{\Phi}$.  In other words, the physically relevant part of $\Phi$ is representable as subsystem dynamics.
\end{mycorollary}
\begin{proof}
	Let $\mathcal{Q} = \Span_{\mathbb{C}}\{\rho\otimes\rho\;:\;\rho\in\Omega_{\Phi}\}$ and let $\mathcal{W}\subset \BHSB$ be the symmetric sector (i.e. the $+1$ eigenspace of the $\Ad_{\textsc{swap}}$ operator on $\mathcal{R}\otimes\mathcal{R}\subset \BHSB \simeq \BHS\otimes\BHS$.  Since $W\in\mathcal{W}$ belongs to $\mathcal{R}\otimes\mathcal{R}$, if $\{A_{i}\}$ is any basis for $\mathcal{R}$, then $W$ can be expanded as $W = \sum_{i,j}a_{ij}A_{i}\otimes A_{j}$.  Since $W\in\mathcal{W}$, $\Ad_{\textsc{swap}}(W) = W$, so that $W = \sum_{i}a_{ii}A_{i}\otimes A_{i} + \frac{1}{2}\sum_{i<j}(a_{ij}+a_{ji})(A_{i}\otimes A_{j} + A_{j}\otimes A_{i})$.  It then becomes clear that $\mathcal{W} = \Span_{\mathbb{C}}\{A_{i}\otimes A_{j} + A_{j}\otimes A_{i}\;:\; i\leq j\}$.  Since any $\rho\otimes\rho$ belongs to $\mathcal{W}$, it follows that $\mathcal{Q}\subset \mathcal{W}$.  Moreover, since $\Omega_{\Phi}$ spans $\mathcal{R}$ by assumption, there exists $\{\rho_{j}\}\subset\Omega_{\Phi}$ and coefficients $\{a_{ij}\}\subset\mathbb{C}$ such that $A_{i} = \sum_{j}a_{ij}\rho_{j}$.  It may also be noted that, since $\Omega_{\Phi}$ is a convex set, if $\rho,\sigma\in\Omega_{\Phi}$, then $(\rho + \sigma)/2\in\Omega_{\Phi}$, so that $\mathcal{Q}$ contains $\rho\otimes\rho$, $\sigma\otimes\sigma$, and $(\rho+\sigma)/2\otimes(\rho+\sigma)/2 = (\rho\otimes\rho + \sigma\otimes\sigma + \rho\otimes\sigma + \sigma\otimes \rho)/4$, and therefore $\mathcal{Q}$ contains $\rho\otimes\sigma + \sigma\otimes\rho$ for all $\rho, \sigma \in\Omega_{\Phi}$.  Then $A_{i}\otimes A_{j} + A_{j}\otimes A_{i} = \sum_{km}a_{ik}a_{jm}(\rho_{k}\otimes\rho_{m} +\rho_{m}\otimes\rho_{k})\in\mathcal{Q}$ for any $i,j$, so that $\mathcal{W}\subset\mathcal{Q}$, and therefore $\mathcal{W} = \mathcal{Q}$ and it is easy to see that $\mathcal{V} = \id\otimes\Phi(\mathcal{Q}) = \id\otimes\Phi(\mathcal{W})$.

	If $\rho\in\TrB[\CS\cap\DSB]$, then $\Psi_{U}^{\CS}(\rho) = \Phi(\rho) \in\DS$, since there exists $\rhoSB\in\CS\cap\DSB$ such that $\rho = \TrB\rhoSB$ and clearly $\Psi_{U}^{\CS}(\rho) = \TrB[U\rhoSB U^{\dag}] \in\DS$.  On the other hand, for this representation, if $\rho\in\mathcal{R}\cap\DS - \Phi^{-1}(\DS)$, then any $X\in\CS$ covering $\rho$ must be $X = \sum_{i}\alpha_{i}\rho_{i}\otimes\Phi(\rho_{i})$ for $\{\alpha_{i}\}\subset\mathbb{C}$ and $\{\rho_{i}\}\subset\Omega_{\Phi}$, where $\rho = \TrB X = \sum_{i}\alpha_{i}\rho_{i}$.  Then $\Tr_{\textsc{s}} X = \sum_{i}\alpha_{i}\Phi(\rho_{i}) = \Phi(\rho) \notin \DS$, since $\rho\notin \Phi^{-1}(\DS)$ by assumption.  Then $X$ cannot belong to $\DSB$, so $\rho\notin\TrB[\CS\cap\DSB]$.  It follows that $\TrB[\CS\cap\DSB] = \Omega_{\Phi}$ for this representation.
\end{proof}

\section{Specifications of Two \textsc{swap}-Consistent Subspaces}
	The \textsc{swap}-consistent subspace
	\begin{equation}
		\CS := \Span_{\mathbb{C}}\{\rho\jmdotimes\rho^{T}\;:\;\rho\in\DS\}
	\end{equation}
	yielding the transpose map is the 10-dimensional subspace
	\begin{align}
		\CS = \Span_{\mathbb{C}}\{&\identity, (X\jmdotimes\identity+\identity\jmdotimes X), (Y\jmdotimes\identity - \identity\jmdotimes Y), \nonumber\\
		& (Z\jmdotimes\identity + \identity\jmdotimes Z), X\jmdotimes X, Y\jmdotimes Y, Z\jmdotimes Z,\nonumber\\
		& (X\jmdotimes Y-Y\jmdotimes X), (Y\jmdotimes Z - Z\jmdotimes Y),\nonumber\\
		& (Z\jmdotimes X + X\jmdotimes Z)\}. 
	\end{align}

	The \textsc{swap}-consistent subspace
	\begin{equation}
		\CS := \Span_{\mathbb{C}}\{\rho\jmdotimes\Phi(\rho)\;:\;\rho\in\Omega_{\Phi}\}
	\end{equation}
	yielding the repolarizing map $\Phi$ is the 10-dimensional subspace	
	\begin{align}
		\CS = \Span_{\mathbb{C}}\Big\{&\identity, \big(X\otimes\identity + \frac{1}{\epsilon}\identity\otimes X\big), \big(Y\otimes\identity + \frac{1}{\epsilon}\identity\otimes Y\big),\nonumber\\
		& \big(Z\otimes\identity + \frac{1}{\epsilon}\identity\otimes Z\big), X\otimes X, Y\otimes Y, Z\otimes Z,\nonumber\\
		& \big(X\otimes Y + Y\otimes X\big), \big(Y\otimes Z + Z\otimes Y\big), \nonumber\\
		& \big(Z\otimes X + X\otimes Z\big)\Big\}
	\end{align}

\end{document}